%% file: Tilted-CHSH.tex
\documentclass[a4paper]{article}

\input{Packages}
\input{Commands}
\input{Layouts}

\title{The tilted CHSH games:
\texorpdfstring{\\}{  }
an operator algebraic classification}
\author{Alexander Frei and Azin Shahiri}
\date{\today}

\begin{document}
\maketitle
\thispagestyle{empty}

\input{Abstract}

\input{Introduction}

\input{Representations}
\input{Binary-games}

\input{Quantum-value}

\input{Quantum-state}

\input{Acknowledgements}

\appendix
\input{Optimisation}

\bibliography{Bibliography}
\bibliographystyle{alpha-all}

\end{document}

%% file: Packages.tex
\usepackage[american]{babel}
\usepackage{amsmath,amssymb,amsthm}
\usepackage{centernot,braket,cancel}
\usepackage{lipsum}
\usepackage{calc}
\usepackage{xcolor}
\usepackage{csquotes}
\usepackage{comment}
\usepackage{fancyhdr}
\usepackage{titling}    
\usepackage{titlesec}   
\usepackage[textwidth=4cm]{todonotes}
\usepackage{hyperref}
\usepackage{nicematrix}

\usepackage{tikz}
\usetikzlibrary{arrows}
\usetikzlibrary{cd}
\usetikzlibrary{calc}
\def\temp{&} \catcode`&=\active \let&=\temp


%% file: Commands.tex
\newcommand{\TAB}{\quad}

\newcommand{\NOT}{\centernot}
\newcommand{\differential}{d}
\newcommand{\CSTAR}{\ensuremath{C^*}}

\DeclareMathSymbol{\mlq}{\mathord}{operators}{'134}
\DeclareMathSymbol{\mrq}{\mathord}{operators}{'42}

\newcommand{\BLANK}{{-}}

\newenvironment{TIKZCD}{\[\begin{tikzcd}}{\end{tikzcd}\]\ignorespacesafterend}

\newcommand{\MATRIX}[1]{\left(\hskip\arraycolsep\begin{matrix}#1\end{matrix}\hskip\arraycolsep\right)}
\newcommand{\mediumMATRIX}[1]{
  \begingroup
  \setlength\arraycolsep{2pt}
  \renewcommand*{\arraystretch}{0.6}
  \left(\hskip\arraycolsep\begin{matrix}#1\end{matrix}\hskip\arraycolsep\right)
  \endgroup}

\newcommand{\smallMATRIX}[1]{\left(\begin{smallmatrix}#1\end{smallmatrix}\right)}

\usepackage{scalerel,stackengine}
\newcommand\medium[1]{\scalerel*[5pt]{\big#1}{%
  \ensurestackMath{\addstackgap[1.5pt]{\big#1}}}}
\newcommand\mediumleft[1]{\mathopen{\medium{#1}}}
\newcommand\mediumright[1]{\mathclose{\medium{#1}}}

\renewcommand{\subset}{\subseteq}

\renewcommand{\phi}{\varphi}
\renewcommand{\epsilon}{\varepsilon}

\makeatletter
\newcommand{\xRightarrow}[2][]{\ext@arrow 0359\Rightarrowfill@{#1}{#2}}
\makeatother

\newcommand{\repr}{\mathrm{repr}}
\newcommand{\irrep}{\mathrm{irrep}}

\newcommand{\INTEGERS}{\mathbb{Z}}

\newcommand{\REALS}{\mathbb{R}}
\newcommand{\COMPLEX}{\mathbb{C}}
\newcommand{\TORUS}{\mathbb{T}}
\DeclareMathOperator{\Realpart}{Re}

\newcommand{\id}{\mathrm{id}}
\newcommand{\inv}{^{-1}}

\DeclareMathOperator{\im}{im}

\DeclareMathOperator{\rank}{rk}

\newcommand{\maxket}{\ket{\max}}

\newcommand{\winning}{\mathrm{winning}}

\newcommand{\aux}{\mathrm{aux}}
\newcommand{\game}{\mathrm{game}}
\newcommand{\bias}{\mathrm{bias}}
\newcommand{\best}{\mathrm{best}}
\newcommand{\rest}{\mathrm{rest}}

\newcommand{\CHSH}{\mathrm{CHSH}}

\newcommand{\Alice}{\mathrm{Alice}}
\newcommand{\Bob}{\mathrm{Bob}}

\newcommand{\twoplayer}{\mathrm{two~player}}
\newcommand{\binaryinput}{\mbox{binary--input}}
\newcommand{\binaryoutput}{\mbox{binary--output}}

\newcommand{\words}{\mathrm{words}}

\newcommand{\spectrum}{\sigma}

\DeclareMathOperator*{\LINEARSPAN}{span}

\DeclareMathOperator{\pos}{pos}

\newcommand{\CONTINUOUS}{C}

\newcommand{\tensor}{\mathbin{\otimes}}

%% file: Layouts.tex
\fancypagestyle{NOPAGENUMBER}{\fancyfoot{}}
\assignpagestyle{\chapter}{empty}
\titlespacing{\chapter}{0pt}{-2\baselineskip}{2\baselineskip}
\titleformat
{\chapter} 
[frame] 
{\Large\bfseries\centering} 
{Part~\thechapter} 
{\baselineskip} 
{\filcenter} 

\usepackage{setspace}
\hypersetup{colorlinks,citecolor=blue,bookmarksnumbered}
\setuptodonotes{color=blue!10,bordercolor=blue!10}

\tikzset{every loop/.style={ distance=1cm, out=120, in=60, -> }}
\tikzcdset{row sep/normal=1cm}
\tikzcdset{column sep/normal=1cm}
\tikzcdset{row sep/small=0.5cm}
\tikzcdset{column sep/small=0.5cm}
\tikzcdset{row sep/large=1.5cm}
\tikzcdset{column sep/large=1.5cm}
\tikzcdset{row sep/huge=2cm}
\tikzcdset{column sep/huge=2cm}
\tikzset{help lines/.style={very thin, color=lightgray, dashed}}
\tikzset{axis lines/.style={very thin, color=lightgray}}

\makeatletter
\def\mathcolor#1#{\@mathcolor{#1}}
\def\@mathcolor#1#2#3{%
  \protect\leavevmode
  \begingroup
    \color#1{#2}#3%
  \endgroup}
\makeatother

\let\originalleft\left
\let\originalright\right
\renewcommand{\left}{\mathopen{}\mathclose\bgroup\originalleft}
\renewcommand{\right}{\aftergroup\egroup\originalright}

\newtheoremstyle{MYBREAK}
  {}          
  {}          
  {\itshape}  
  {}          
  {\bfseries} 
  {:}         
  {\newline}  
  {}          
\newtheoremstyle{MYPLAIN}
  {\topsep}   
  {\topsep}   
  {\itshape}  
  {15pt}          
  {\bfseries} 
  {}         
  {5pt plus 1pt minus 1pt} 
  {}          

\theoremstyle{definition}
\newtheorem{DEF}{Definition}[section]

\theoremstyle{plain}
\newtheorem{PROP}[DEF]{Proposition}

\newtheorem{THM}[DEF]{Theorem}

\newtheorem{INTROTHM}{Theorem}
\newtheorem*{INTROPROP}{Proposition}

\theoremstyle{MYBREAK}
\newtheorem{BREAKPROP}[DEF]{Proposition}

\theoremstyle{MYBREAK}
\newtheorem*{NONUMBER-BREAKTHM}{Theorem}
\newtheorem*{NONUMBER-BREAKCOR}{Corollary}

\theoremstyle{plain}
\newtheorem*{NONUMBER-THM}{Theorem}


%% file: Abstract.tex
\begin{abstract}
  We introduce a general systematic procedure for solving any\linebreak
  \mbox{binary--input} \mbox{binary--output} game using operator algebraic techniques on the representation theory for the underlying group,
  which we then illustrate on the prominent class of tilted CHSH games:\\
  We derive for those an entire characterisation on the region exhibiting some quantum advantage and in particular derive a greatly simplified description for the required amount of anticommutation on observables (as being an essential ingredient in several adjacent articles).\\
  We further derive an abstract algebraic representation--free classification on the unique operator algebraic state maximising above quantum value.\linebreak
  In particular the resulting operator algebraic state entails uniqueness for its corresponding correlation, including all higher and mixed moments.\\
  Finally the main purpose of this article is to provide above simplified description for the required amount of anticommutation and an abstract algebraic characterisation for their corresponding unique optimal state, both defining a key ingredient in upcoming work by the authors.
\end{abstract}

%% file: Introduction.tex
\section*{Introduction}

Our motivation stems from the recently obtained important separation between quantum spatial and finite dimensional correlations as obtained in \cite{COLADANGELO-STARK}
\[
  p(\BLANK|\BLANK)\in C_{qs}(A,B|X,Y)
  \TAB\text{whereas}\TAB
  p(\BLANK|\BLANK)\notin C_{q}(A,B|X,Y)
\]
which builds on the following essential phenomena from tilted CHSH games:\linebreak
As a first phenomenon the tilted CHSH games are known to admit a unique optimal strategy at least for up to local dilation of the form
\[
  \mediumMATRIX{u'(x)\\[2\jot]&0}\tensor\mediumMATRIX{1\\&0}\mediumMATRIX{\ket{\phi'}\\[2\jot]0}
  \TAB\text{and}\TAB
  \mediumMATRIX{1\\&0}\tensor\mediumMATRIX{v(y)\\[2\jot]&0}\mediumMATRIX{\ket{\phi'}\\[2\jot]0}
\]
with matrix notation as in \cite[section~1]{FREI-CHSH}
and auxiliary environment
\[
  \Bigl(u(x)\tensor1\Bigr)\ket{\phi}\tensor\ket{\aux}
  \TAB\text{and}\TAB
  \Bigl(1\tensor v(y)\Bigr)\ket{\phi}\tensor\ket{\aux}.
\]
However as a drawback of this known self-testing result, either of above freedom in choice of representation revokes uniqueness of optimal states along higher moments as the following implication fails in general
\[
  u(x)\ket\phi=u'(x)\ket{\phi'}
  \TAB\NOT\implies\TAB
  u(x_1)\cdots u(x_n)\ket\phi=u'(x_1)\cdots u'(x_n)\ket{\phi'}
\]
while some partial result holds in certain cases as illustrated in \cite{PADDOCK-SLOFSTRA-ZHAO-ZHOU}.\\
Secondly the ideal strategy further arises by some partially entangled state
\[
  \ket\phi=\cos(\phi)\ket{0}\tensor\ket{0} +\sin(\phi)\ket1\tensor\ket1:\TAB\TAB
  \sin(2\phi)=\sqrt{\frac{1-\alpha^2\beta^2}{1+\beta^2}}\hspace{-1cm}
\]
with additional scaling on the second parameter for convenience as in \eqref{DEFINITION:TILTED-CHSH}.\\
The corresponding observables are thereby usually formulated as a particular representation written in the form
\begin{align*}
  u\tensor1=Z\tensor1,&\TAB\TAB\TAB 1\tensor u=1\tensor\Bigl(\cos(\mu) Z+\sin(\mu) X\Bigr)\hspace{-1cm}\\[2\jot]
  v\tensor1=X\tensor1,&\TAB\TAB\TAB 1\tensor v=1\tensor\Bigl(\cos(\mu) Z-\sin(\mu) X\Bigr)\hspace{-1cm}
\end{align*}
for some angle formulated as (for the case $\alpha=0$)
\[
  \tan(\mu)=\sqrt{\frac{1-\beta^2}{1+\beta^2}}
  \TAB\implies\TAB
  \mediumleft(~\cos(\mu)={?}~\medium|~\sin(\mu)={?}~\mediumright)\hspace{-1cm}
\]
as to be found in the original article \cite{ACIN-MASSAR-PIRONIO} as well as further in \cite{BAMPS-PIRONIO}\\
which entails an additional burden on unraveling the given formula.\\
Using these two phenomena led the authors to construct quantum spatial correlations that do not admit any finite dimensional representation.

While the separation from \cite{COLADANGELO-STARK} answers a long standing open problem,\linebreak
the authors of the current article believe this separation to be of larger more general operator algebraic phenomenon.
In order to solve this problem, the authors thus aimed as a crucial first task to get a deeper and better understanding on optimal states for tilted CHSH games from an operator algebraic perspective,
which defines the main motivation for the current work.

For this the authors employ a novel more systematic procedure for nonlocal games as initiated by the first named author in the introductory article \cite{FREI-CHSH}:\\
More precisely, the crucial ingredient arises from the canonical decomposition for states on operator algebras into a representation for the operator algebra followed by some vector state on the representation
\begin{TIKZCD}[column sep=large]
  \CSTAR(\twoplayer)\rar{\repr}& Q\subset B(H)\rar{\bra\phi\BLANK\ket\phi}&\COMPLEX
\end{TIKZCD}
with notation for the twoplayer algebra from \cite[section~2]{FREI-CHSH}.\\
This canonical decomposition is uniquely determined (and arising as their minimal Stinespring dilation) and may be seen in analogy with the canonical decomposition of morphisms into their epimorphisms followed by their\linebreak monomorphism (for more details confer \cite[section~1]{FREI-CHSH}).\\
As such the question on characterising optimal states splits into the first task of finding the common minimal quotient hosting the entire optimal state space for the given nonlocal game
\[
  \begin{tikzcd}[column sep=small]
    \CSTAR(\twoplayer)\rar&\ldots\rar& Q(\game):
  \end{tikzcd}
  \TAB\TAB
  S(\game)\subset SQ(\game)
\]
followed by the second task of characterising its optimal states as a subspace.\\
As a consequence one derives a general obstruction for optimal representations (given by the minimal quotient) in the form of polynomial relations such as
\begin{gather*}
  \begin{tikzcd}
    \pi:\CSTAR(\twoplayer)\rar& Q(\game)\rar& B(H):
  \end{tikzcd}\\[2\jot]
  \hspace{-1cm}\implies\TAB\pi\biggl(\sqrt{2}E(2|1)E(0|3)-5E(1|1)\biggr)=0
\end{gather*}
together with a characterisation of optimal states including all moments
\begin{gather*}
  \phi(1\tensor1)=1,\TAB \phi\Bigl(E(a|x)\tensor 1\Bigr),\TAB
  \phi\Bigl(1\tensor E(b|y)\Bigr),\TAB
  \phi\Bigl(E(a|x)\tensor E(b|y)\Bigr) \\[2\jot]
  \phi\Bigl(E(a|x)E(a'|x')\tensor 1\Bigr),\TAB
  \phi\Bigl(E(a|x)E(a'|x')\tensor E(b|y)\Bigr),\TAB\ldots
\end{gather*}
This novel procedure has been illustrated on a variety of examples in \cite{FREI-CHSH}.\\
In the current article we devote this general procedure entirely on the tilted CHSH games as a fully fletched example from start to finish.

For this we revisit a classical result from operator algebras on the univeral algebra generated by a pair of projections.
Remodeling said result in terms of generating unitaries (for the underlying group) in combination with induction on the trivial representation (as basically the canonical representation for Fell bundles in general) reveals a surprising description as the algebra generated by unitaries along their range of anticommutation:
\begin{INTROTHM}\label{INTRO-THEOREM:REPRESENTATIONS}
  The universal algebra generated by a pair of order--two unitaries admits the algebraic description as function algebra ranging over the amount of anticommutation for its generators
  \[
    \CSTAR\mediumleft(~u^2=1=v^2~\mediumright) =
    \CSTAR\Bigr(~u^2=1=v^2~\Big|~\{u,v\}=2\id~\Bigr) \subset
    \CONTINUOUS\mediumleft([-1,1]\to M_2\mediumright)
  \]
  for any choice of continuous generators $u,v\in\CONTINUOUS\mediumleft([-1,1],M_2\mediumright)$.
\end{INTROTHM}

Our investigations further reveal this relation as the underpinning of strategies from most of adjacent articles such as \cite{PAL-VERTESI} on the prominent I3322 inequality from \cite{FROISSART-I3322} which the authors elaborate in upcoming work.

As such the computation of quantum values for \binaryinput\ \binaryoutput\ games reduces to some optimisation problem within matrix algebras
\[
  \game(s,t)\in M_2\tensor M_2:\TAB\TAB\|\game\|=\sup_{-1\leq s,t\leq1}\|\game(s,t)\|
\]
along the amount of anticommutation between generators
\[
  M_2\tensor M_2=\CSTAR\Bigr(\{u,v\}=2s(\Alice)\Bigr)\tensor\CSTAR\Bigr(\{u,v\}=2s(\Bob)\Bigr).
\]
For the desired tilted CHSH games, we then derive using this method the familiar quantum value,
for which we however additionally reveal the precisely required amount of anticommutation between generators:
\begin{INTROTHM}The tilted CHSH games have maximal quantum value (in some precisely specified region as given in detail in theorem \ref{THEOREM:TILTED:QUANTUM-VALUE})
  \begin{gather*}
    \|\CHSH(\alpha,\beta)\|^2=4(1+\alpha^2)(1+\beta^2)
  \end{gather*}
only along the following amount of anticommutation between generators
\[
  \left(~s(\Alice)=0~\Big|~s(\Bob)=\frac{\alpha^2-1}{\alpha^2+1}(\beta^2+1) + \beta^2~\right).
\]
As such this fully characterises the obstruction on representations (see above).
\end{INTROTHM}
In particular this reveals a simplified formula for their generators:\\
More precisely we note that the freedom in the usual representation as above (or in more physical terms as a choice of reference frame)
\[
  1\tensor u=1\tensor\Bigl(\cos(\alpha) Z+\sin(\alpha) X\Bigr)
  \TAB\text{and}\TAB
  1\tensor v=1\tensor\Bigl(\cos(\beta) Z+\sin(\beta) X\Bigr)
\]
is after all uniquely characterised in terms of their relative angle
\[
  \{1\tensor u,1\tensor v\}= \cos(\alpha)\cos(\beta)-\sin(\alpha)\sin(\beta) = \cos|\alpha-\beta|=2s(\Bob)
\]
for which we thus delivered a convenient formulation by our result above.\\
This includes the usual representation as outlined at the beginning.
As such this solves the first task
given by the canonical decomposition of operator algebraic states
into representations and vector states.

In order to classify their optimal states we further provide the following result
(which certainly appears scattered throughout literature in similar forms):
\begin{INTROPROP}
  For any abritrary element (seen as a bias polynomial)
  and with operator norm (seen as its maximal or rather extremal bias)
  \[
    \bias\in\CONTINUOUS_0(X,A):\TAB\TAB \|\bias\|=\sup\|\bias(X)\|
  \]
  consider the corresponding short exact sequence
  \begin{TIKZCD}
    0\rar& \CONTINUOUS_0(\rest,A) \rar& \CONTINUOUS_0(X,A) \rar &\CONTINUOUS_0(\best,A)\rar& 0
  \end{TIKZCD}
  given by the defining subspace (and its complement)
  \[
    \best:=\Bigl\{~\|\bias(x)\|=\|\bias\|~\Bigr\}\subset X.
  \]
  Then any operator algebraic state $\phi:\CONTINUOUS_0(X,A)\to\COMPLEX$
  maximising above bias in the sense $|\phi(\bias)|=\|\bias\|$
  necessarily factors by the quotient
  \begin{TIKZCD}
    \phi:\CONTINUOUS_0(X,A)
    \rar&
    \CONTINUOUS_0(\best,A)
    \rar&
    \COMPLEX
  \end{TIKZCD}
  and as such the optimal state space arises as a subspace for the quotient.
\end{INTROPROP}

With theorem \ref{INTRO-THEOREM:REPRESENTATIONS} in mind, we may thus in principle solve every \binaryinput\ \binaryoutput\ game. In particular for those with above quotient algebra given by a single parameter pair such as
\begin{gather*}
  \best=\left(~s(\Alice)=0~\Big|~s(\Bob)=\frac{\alpha^2-1}{\alpha^2+1}(\beta^2+1) + \beta^2~\right):\\[2\jot]
  \CSTAR\Bigr(\{u,v\}=2s(\Alice)\Bigr)\tensor\CSTAR\Bigr(\{u,v\}=2s(\Bob)\Bigr)= M_2\tensor M_2
\end{gather*}
the solution space follows as the convex closure by the eigenspace problem
\[
  \bias\in M_2\tensor M_2:\TAB\TAB \bias\ket\phi=\|\bias\|\cdot\ket\phi~?
\]
while more generally the solution space arises as continuous sections thereof.\\
With these reductions we finally classify the entire solution space (including all higher moments) for the tilted CHSH games as:

\begin{INTROTHM}
  The tilted CHSH games admit unique optimal states as uniquely determined by the correlation table
  (which further entails all higher moments using the anticommutation relation on generators):
  \[
    \def\arraystretch{1.5}
    \begin{array}[c]{c|c|c|c|c}
      \phi(\ldots) & \BLANK\tensor 1 & \BLANK\tensor 2u & \BLANK\tensor 2v & \BLANK\tensor 4uv \\
      \hline
      1\tensor\BLANK & 1=a^2+d^2 & (a^2-d^2)|w_+| & (a^2-d^2)|w_+| & (a^2+d^2)(|w_+|^2-|w_-|^2) \\
      \hline
      u\tensor\BLANK & a^2-d^2 & (a^2+d^2)|w_+| & (a^2+d^2)|w_+| & (a^2-d^2)(|w_+|^2-|w_-|^2) \\
      \hline
      v\tensor\BLANK & 0 & 2ad|w_-| & -2ad|w_-| & 0 \\
      \hline
      uv\tensor\BLANK & 0 & 0 & 0 & 4ad|w_+|\cdot|w_-|
    \end{array}
  \]
  The values are thereby given by the familiar transformation
  \[
    \hspace{-1cm}w_\pm=u\pm v:\TAB\TAB|w_+|=2\alpha\sqrt{\frac{1+\beta^2}{1+\alpha^2}}
    \TAB\text{and}\TAB
    |w_-|=2\sqrt{\frac{1-\alpha^2\beta^2}{1+\alpha^2}}
  \]
  as well as Schmidt coefficients given by
  \[
    a=\sqrt{\frac{\sqrt{1+\beta^2}+\beta\sqrt{1+\alpha^2}}{2\sqrt{1+\beta^2}}}
    \TAB\text{and}\TAB
    d=\sqrt{\frac{\sqrt{1+\beta^2}-\beta\sqrt{1+\alpha^2}}{2\sqrt{1+\beta^2}}}\hspace{-1.5cm}
  \]
  from which there is no ambiguity left anymore.\\
  For more details on above optimal state confer theorem \ref{THEOREM:TILTED:QUANTUM-STATE}.
\end{INTROTHM}

We note that this characterisation is much stronger than any previous self-testing result as it entails uniqueness for states including all higher and mixed moments (see the argument in the beginning of this introduction).\\
Furthermore, the usual form (as in the beginning)
\begin{gather*}
  \ket\phi=\cos(\phi)\ket{0}\tensor\ket{0} +\sin(\phi)\ket1\tensor\ket1
\end{gather*}
simply reads off as above Schmidt coefficients
\[
  \left(~~\cos(\phi)=\sqrt{\frac{\sqrt{1+\beta^2}+\beta\sqrt{1+\alpha^2}}{2\sqrt{1+\beta^2}}}
  ~\Big|~
  \sin(\phi)=\sqrt{\frac{\sqrt{1+\beta^2}-\beta\sqrt{1+\alpha^2}}{2\sqrt{1+\beta^2}}}~~\right).
\]
This compares precisely to the usual formulation using
\[
  \sin(2\phi)^2=4\sin^2(\phi)\cos^2(\phi)=\left(\sqrt{\frac{1-\alpha^2\beta^2}{1+\beta^2}}\right)^2
\]
which however merely defines an implicit formulation (see the beginning).\\
As such we further derived an explicit concrete formulation for their Schmidt coefficients with no necessity for any trigonmetric formulation.

Concluding this introduction,
our current work provides a systematic procedure for solving any \binaryinput\ \binaryoutput\ game,
and in particular illustrates this procedure as an entire solution for the tilted CHSH games:\\
This includes a precise characterisation for the entire region of quantum advantage (with additional phase separation which remained unnoticed) as well as the characterisation for its required amount of anticommutation as an obstruction on representations.
On the other hand this entails an entire classification on optimal states as a uniquely determined operator algebraic state (including all higher moments) and as such diminishes every ambiguity that remained from traditional self-testing results for the tilted CHSH games.

%% file: Representations.tex
\section[Representation-theory]{Representation-theory:\texorpdfstring{\\}{  }an algebraic description}
\label{sec:REPRESENTATION-THEORY}

The beginning of our investigation starts with a classical result from operator algebras from a slightly different perspective (which matches astonishingly well with quantum information theoretic consideration).
More precisely, the classical result gives a concrete description for the universal operator algebra generated by a pair of projections as some matrix-valued function algebra
\[
  \CSTAR\Bigl(p^2=p=p^*,q^2=q=q^*\Bigr)\subset\CONTINUOUS\Bigl(\mathrm{some~space},M_2\Bigr).
\]
While those descriptions usually involve some rather concrete interpretations such as rotations (as suggested for instance in \cite[remark 1.4]{RAEBURN-SINCLAIR})
\[
  p=\mediumMATRIX{1\\0}\mediumMATRIX{1&0}\TAB\text{and}\TAB q=\mediumMATRIX{\cos\\\sin}\mediumMATRIX{\cos&\sin}
\]
they generally lack for some description using abstract algebraic relations.
Using instead the perspective of generating unitaries
we will provide such as the range of possible anticommutation
\[
  \{u,v\}=uv+vu=2s,\TAB\text{for}\TAB -1\leq s\leq1.
\]
For their corresponding projections however this anticommutation reads when written out as the meaningless relation
\[
  \{2p-1,2q-1\}=4(pq+qp) - 4(p+q) + 2=2s
\]
which may be the reason why above algebraic description remained overlooked.
Further, the previous methods on above identification as matrix-valued functions
usually involve some counting argument on the irreducible\linebreak representations
for the underlying group (as a semidirect product)
\[
  \INTEGERS/2*\INTEGERS/2=\INTEGERS\rtimes\INTEGERS/2.
\]
As the underlying group however arises as an amenable semidirect product,
induction on representations becomes the viable technique
and as such provides a canonical construction
(in fact the most canonical one on Fell bundles):

\begin{PROP}The semidirect product given by inversion
  \[
    \INTEGERS\rtimes\INTEGERS/2=\mediumleft\langle~w,s~\medium|~sws=w\inv~\text{and}~s^2=1~\mediumright\rangle
  \]
  defines the canonical embedding
  $\CSTAR\mediumleft(\INTEGERS\rtimes\INTEGERS/2\mediumright)\hookrightarrow \CSTAR(\INTEGERS)\tensor M(2)$
  \begin{gather*}
    w\rtimes1=\mediumMATRIX{w\\&w\inv}\TAB\text{and}\TAB1\rtimes s=\mediumMATRIX{&1\\1}
  \end{gather*}
  resembling the regular representation (up to some change of basis).
\end{PROP}
\begin{proof}We begin in the more general context of crossed products for group actions on operator algebras:
  For these recall that induction on representations alters the coefficient algebra as diagonal operators (while keeping the group acting as bilateral shift)
  \[
    A\rtimes G\to B\mediumleft(\ell^2G\tensor A\mediumright):\TAB
    1\rtimes g\mapsto 1\tensor g
    \TAB\text{and}\TAB
    a\rtimes1\mapsto\smallMATRIX{\ldots\\&g^*ag\\&&\ldots}.
  \]
  The tensor product here is more precisely given as a Hilbert right module,\\
  which basically resembles a skewed formulation for the regular representation.
  Further this comes equipped with a conditional expectation satisfying
  \[
    \bra{e}\BLANK\ket{e}: A\rtimes G\to A:\TAB\TAB \bra{e}\BLANK\rtimes e\ket{e}=\id_A(\BLANK)
    \TAB\text{and}\TAB
    \bra{e}\BLANK\rtimes(G\setminus e)\ket{e}=0
  \]
  which are precisely the relations defining the regular one $E: A\rtimes G\to A$.\\
  As such we have found nothing but an intermediate quotient
  \begin{TIKZCD}
    A\rtimes G \rar& \im(A\rtimes G) \rar& A\rtimes G\medium/\mediumleft(E:A\rtimes G\to A\mediumright)
  \end{TIKZCD}
  where we confer to
  \cite[section~2.1]{KWASNIEWSKI-EXEL-CROSSED} and \cite[section~3.1]{KWASNIEWSKI-MEYER-ESSENTIAL-CROSSED}\\
  in combination with \cite[section~1]{FREI-CHSH} and further \cite[section~17]{Exel-BOOK-2017}.\\
  In fact above representation resembles the far right regular quotient.\\
  For amenable groups (such as ours) these however all agree
  \begin{TIKZCD}
    A\rtimes G \rar[equal]& \im(A\rtimes G) \rar[equal]& A\rtimes G\medium/\mediumleft(E:A\rtimes G\to A\mediumright)
  \end{TIKZCD}
  and whence each defines a faithful representation for the universal one.\\
  Now for finite groups above representation may be already taken as
  \begin{TIKZCD}
    A\rtimes \mediumleft(F=\mathrm{finite}\mediumright)
    \rar[hookrightarrow]&
    B\mediumleft(\ell^2 F\mediumright)\tensor A
    \rar[hookrightarrow]&
    B\mediumleft(\ell^2F\tensor A\mediumright)
  \end{TIKZCD}
  which basically defines a matrix algebra over the given coefficient algebra.\\
  On the other hand for semidirect products both their group algebra and their crossed product define one and the same categorical object and whence altogher
  \begin{TIKZCD}
    \CSTAR(\INTEGERS\rtimes\INTEGERS/2)
    \rar[equal]&
    \CSTAR(\INTEGERS)\rtimes\INTEGERS/2
    \rar[hookrightarrow]&
    M(2)\tensor\CSTAR(\INTEGERS)
  \end{TIKZCD}
  whose generators read written out
  \[
    w\rtimes1=\mediumMATRIX{ewe\\&sws}=\mediumMATRIX{w\\&w\inv}\TAB\text{and}\TAB 1\rtimes s=\mediumMATRIX{&1\\1}.
  \]
  This gives the desired conclusion.
\end{proof}

We may now employ the previous result in case of our presentation above:
More precisely we recall the well-known isomorphism given for instance as (the particular choice will drop out in our final result below)
\begin{gather*}
  \INTEGERS/2*\INTEGERS/2 = \mediumleft\langle u^2=1=v^2\mediumright\rangle =
  \mediumleft\langle sws=w\inv,s^2=1\mediumright\rangle=\INTEGERS\rtimes\INTEGERS/2:\\[2\jot]
  \text{e.g.}\TAB(~u=ws,~v=s~)\iff (~w=uv,~s=v~).
\end{gather*}
In fact this isomorphism defines a particular instance of more generally freely generated subgroups of the form (as already arising in the I3322 inequality)
\begin{gather*}
  \mediumleft(\INTEGERS*\ldots*\INTEGERS\mediumright)=\mediumleft\langle~w_1=u_1v,~\ldots~,~w_n=u_nv~\mediumright\rangle
  \\[2\jot]\hspace{2cm}\subset
  \mediumleft\langle~u_1^2,\ldots,u_n^2=1,~v^2=1~\mediumright\rangle =
  \mediumleft(\INTEGERS/2*\ldots*\INTEGERS/2\mediumright)*\INTEGERS/2.
\end{gather*}
On the other hand we recall the basic identification as circle algebra
\[
  \CSTAR(\INTEGERS)=\CSTAR(w^*w=1=ww^*)=\CONTINUOUS(\TORUS):\TAB w(z)=z
\]
along with the familiar identification (defining an isomorphism)
\[
  \CONTINUOUS(\TORUS)\tensor A=\CONTINUOUS\mediumleft(\TORUS\to A\mediumright):\TAB (f\tensor a)(z)=f(z)\tensor a.
\]
Paired together we thus obtain the identification (so far with choice above)
\begin{gather*}
  \CSTAR(\INTEGERS/2*\INTEGERS/2)=\CSTAR(\INTEGERS)\rtimes\INTEGERS/2\subset\CONTINUOUS(\TORUS,M_2):\\[2\jot]
  u(z)=\mediumMATRIX{z\\&z^*}\mediumMATRIX{&1\\1}\TAB\text{and}\TAB v(z)=\mediumMATRIX{&1\\1}.
\end{gather*}
While this defines some isomorphism for an identification as a function algebra (up to its automorphism group)
this description does by no means define an intrinsic algebraic characterization.
Our goal is now to remedy this lack by identifying their defining algebraic relations.
For this we may now compute their anticommutation (using for instance our tentative choice above)
\begin{align*}
  \mediumleft\{u(z),v(z)\mediumright\} &=
  \mediumMATRIX{z\\&z^*}\cancel{\mediumMATRIX{&1\\1}}\cancel{\mediumMATRIX{&1\\1}} +
  \mediumMATRIX{&1\\1}\mediumMATRIX{z\\&z^*}\mediumMATRIX{&1\\1} \\[2\jot]
  &=\mediumMATRIX{z\\&z^*} + \mediumMATRIX{z^*\\&z} = 2\Realpart(z)\mediumMATRIX{1\\&1}.
\end{align*}
With these preliminary ideas in mind our desired characterisation reads:

\begin{THM}[Abstract characterisation]
  \label{THEOREM:REPRESENTATION-THEORY}
 Our identification above may be altogether characterised as an abstract algebra given by the algebraic relation (for any choice of continuous generators within the function algebra)
 \begin{align*}
   \CSTAR\mediumleft(~u^2=1=v^2~\mediumright)
   =\CSTAR\mediumleft(~u,v~\medium|~\{u,v\}=2\id~\mediumright)
   \subset\CONTINUOUS\mediumleft(~\Realpart(\TORUS)=[-1,1]\to M_2~\mediumright).
 \end{align*}
 In particular its irreducible representations may be parametrised by the range of above possible anticommutations (running over $-1\leq s\leq1$)
 \[
   M(2)=\CSTAR\mediumleft(~u,v~\medium|~\{u(s),v(s)\}=2s~\mediumright)
 \]
 with abstract matrix algebra similar as in \cite[appendix~A]{FREI-CHSH}.
\end{THM}
\begin{proof}
  We wish to verify the following:
  Each such representation given by some choice of continuous generators satisfying above relations defines an isometric embedding which we may equivalently phrase as an embedding into
  \begin{TIKZCD}
    \CSTAR\mediumleft(~u^2=1=v^2~\mediumright)
    \rar[hookrightarrow]&
    \displaystyle\prod_{-1\leq s\leq 1}\CSTAR\mediumleft(~u,v~\medium|~\{u,v\}=2s~\mediumright)\TAB?
  \end{TIKZCD}
  We have however already found one such by our particular choice above,\\
  and as such also every such defines an isometric embedding as desired.
\end{proof}

Before we finish this section let us note how to retrieve usual Pauli matrices:
Given for instance some pair of the form
\[
  u=\cos(\phi)Z+ \sin(\phi)X,\TAB\TAB
  v=\cos(\psi)Z+ \sin(\psi)X
\]
then their anticommutation read off as
\[
  \{u,v\}= 2\cos(\phi)\cos(\psi)+2\sin(\phi)\sin(\psi)=2\cos(\phi-\psi)
\]
and so all that matters is their relative angle.
On the other hand given some pair with certain anticommutation we may separate their overlap as
\[
  \{u,v\}=2\cos(\delta)\TAB\implies\TAB\mediumleft\{u,v-\cos(\delta)u\mediumright\}=\{u,v\}-\cos(\delta)\{u,u\}=0
\]
and further pick to diagonalise either of the generators such as
\begin{gather*}
  u=\cos(0)Z+\sin(0)X=Z\TAB\implies\TAB
  v=\cos(\delta)Z+\sin(\delta)X.
\end{gather*}
We won't be needing however any of above Pauli matrices.
All what really matters is the amount of anticommutation between generators.
With this being said we now move on to nonlocal games.

%% file: Binary-games.tex
\section[Binary--input binary--output games]{Binary--input binary--output games:\\general solution procedure}
With our previous characterisation on the representation theory we may now in principle solve any binary--input binary--output game:
More precisely any binary--input binary output game has as underlying twoplayer algebra (confer \cite[section~2]{FREI-CHSH})
\[
  \CSTAR(\twoplayer)=\CSTAR\mediumleft(u^2=1=v^2\mediumright)\tensor\CSTAR\mediumleft(u^2=1=v^2\mediumright)
\]
and may be characterised as a polynomial such as
\[
  \game= 2\: u\tensor1 - 1\tensor v + \sqrt{3}\: uvu\tensor v + 5\: v\tensor uvuv.
\]
Generally, the investigation of nonlocal games may be broken up into the following pair of tasks (each involving their own kind of special techniques):
(1) compute the operator norm for the given polynomial (given the polynomial has positive or symmetric spectrum)
\[
  \Bigl(\game\geq0:\TAB\|\game\|=?\Bigr)
  \TAB\text{or}\TAB
  \Bigl(\sigma(\bias)=-\sigma(\bias):\TAB \|\bias\|=?\Bigr)
\]
and (2) characterise the space of states maximising the above value
\[
  \phi:\CSTAR(\twoplayer)\to\COMPLEX:\TAB\TAB\phi(\game)=\|\game\|~?
\]
While either or both are generally hard to derive algorithmically (from a complexity point of view) and further are not independent in certain sense,
they may be still either or both solved for particular polynomials given in practice.

Now in order to solve the first task it usually becomes relevant to have profound knowledge on the representation theory for the underlying operator algebra.
For instance a good handle on its space of irreducible representations may be used to derive the operator norm for the given polynomial as an optimisation problem
\[
  \CSTAR(\twoplayer)\subset\prod_{\pi=\mathrm{irrep}}B(H_\pi):\TAB\TAB\|\game\|=\sup_{\pi=\irrep}\|\pi(\game)\|.
\]
Our previous section precisely provides such in the form (see theorem \ref{THEOREM:REPRESENTATION-THEORY})
\begin{align*}
  \CSTAR(\twoplayer) &=
  \CSTAR\mediumright(u^2=1=v^2\mediumright)
  \tensor
  \CSTAR\mediumright(u^2=1=v^2\mediumright)
  \\ &\subset
  \CSTAR\mediumleft(~u,v~\medium|~\{u,v\}=2\id~\mediumright)
  \tensor
  \CSTAR\mediumleft(~u,v~\medium|~\{u,v\}=2\id~\mediumright)
  \\ &\subset
  \CONTINUOUS\mediumleft([-1,1]\times[-1,1],M_2\tensor M_2\mediumright)
\end{align*}
(Note that the spatial and maximal tensor product coincide as both define amenable groups.)
As such we may in principle compute any operator norm for binary--input binary--output games as an optimisation problem along the anticommutation between generators for either player in the sense
\[
  \|\game\|=\sup_{-1\leq s,t\leq1}
  \Bigl\{~\mediumleft\|\game(s,t)\mediumright\|~:~\{u\tensor1,v\tensor1\}=2s,\{u\tensor1,v\tensor1\}=2t~\Bigr\}.
\]
Moreover since their generators anticommute up to scalar constant each monomial reduces to words of at most order--two such as
\[
  \{u,v\}=2s\TAB\implies\TAB vu = -uv + 2s,\TAB\TAB uvu = - v + 2su,\TAB\ldots
\]
and as such every polynomial (in respective matrix algebra) reduces to some linear expression of the form
\begin{gather*}
  \game(s,t) \in \LINEARSPAN\mediumleft\{1,u,v,uv\mediumright\}\tensor\mediumleft\{1,u,v,uv\mediumright\}:\\[2\jot]
  \game(s,t)= \lambda(1,1)1\tensor1 + \lambda(1,u) 1\tensor u + \ldots +\lambda(uv,uv)\tensor uv.
\end{gather*}
One may further (if one feels more comfortable)
replace each abstract generator as a Pauli matrix as outlined in the previous section.
As such the first task reads (for every pair of anticommutation) as an eigenvalue problem within
\begin{gather*}
  \CSTAR\mediumleft(~u,v~\medium|~\{u,v\}=2s~\mediumright)
  \tensor
  \CSTAR\mediumleft(~u,v~\medium|~\{u,v\}=2t~\mediumright)
  = M_2\tensor M_2:\\[2\jot]
  \mediumleft\|\game(s,t)\mediumright\| =
  \mediumleft\|\spectrum\mediumleft(\game(s,t)\mediumright)\mediumright\| =
  r\mediumleft(\game(s,t)\mediumright) =~?
\end{gather*}
which one may solve in principle by hand and from which one may in principle solve the optimisation problem along their anticommutation as above.

In order to solve task (2) to characterise the space of optimal states, one may now invoke the previous knowledge from the optimisation problem:
More precisely consider any state extended by Hahn-Banach to the ambient function algebra
\begin{TIKZCD}[column sep=3cm,row sep=small]
  \CONTINUOUS\mediumleft([-1,1]\times[-1,1]\to M_2\tensor M_2\mediumright)
  \drar[dashed,start anchor=east]\\
  \CSTAR\mediumright(u^2=1=v^2\mediumright)
  \tensor
  \CSTAR\mediumright(u^2=1=v^2\mediumright)
  \uar\rar&
  \COMPLEX.
\end{TIKZCD}
On the other hand, while optimising along the anticommutation we additionally infered those pairs maximising the operator norm
\[
  \best:=\Bigl\{~(s,t)~:~\|\game(s,t)\|=\|\game\|~\Bigr\}\subset [-1,1]\times [-1,1].
\]
As such only those states descending on this maximising subspace may be possibly maximising the operator norm
\begin{TIKZCD}
  \CONTINUOUS\mediumleft([-1,1]\times[-1,1]\to M_2\tensor M_2\mediumright) \rar&
  \CONTINUOUS\mediumleft(\best\to M_2\tensor M_2\mediumright)\rar[dashed]&
  \COMPLEX
\end{TIKZCD}
where these need to further arise as some mixture by the maximal eigenspace.\\
So far about the procedure for solving binary--input binary--output games.

\section{Real polynomials: binary games}

Before continuing we introduce a simplified translation for associating binary games for polynomials of order--two with real possible negative coefficients:
\begin{gather*}
  \bias \in\sum\REALS\cdot\{1,u_1,\ldots,u_n\}\tensor\{1,u_1,\ldots,u_n\}:\\[2\jot]
  \bias = \lambda(1,1)1\tensor1 + \lambda(1,u_1) 1\tensor u_1 + \ldots +\lambda(u_n,u_n) u_n\tensor u_n.
\end{gather*}
For instance the upcoming tilted CHSH games are given as a real polynomial
\[
  \CHSH(\alpha,\beta)=\alpha u\tensor (u+v) + v\tensor (u-v) + 2\beta u\tensor1
\]
which we basically view as a bias polynomial more precisely as follows:\\
For this we first note that any XOR-game may be equivalently described in terms of its bias polynomial (see for instance \cite[section~5.3]{CLEVE-HOYER-TONER-WATROUS})
\[
  \game=\frac{1}{|X\times Y|}\sum E(\winning)
  \TAB\implies\TAB
  \bias=\sum(-1)^{V(x,y)}u(x)\tensor u(y)
\]
where we have used the shorthand notation as in \cite[appendix~B]{FREI-CHSH}.\\
While this defines a pretty much standard procedure, we now describe a simplified procedure for associating nonlocal games to any bias polynomial with real possibly negative coefficients (in a somewhat canonical way):\\
For this we begin with the observation that tensor products of order--two unitaries remain order--two
whose induced spectral decomposition reads for
\begin{gather*}
  \mediumleft(~u = p-q~\mediumright)\tensor1\TAB\text{and}\TAB 1\tensor \mediumleft(~u' = p'-q'~\mediumright):\\[3\jot]
  \implies u\tensor u'= \mediumleft(~p\tensor p' + q\tensor q'~\mediumright) - \mediumleft(~p\tensor q' + q\tensor p'~\mediumright).
\end{gather*}
As such its relations for recoving spectral projections read
\begin{align*}
  \mediumleft(~1\tensor 1 + u\tensor u'~\mediumright)
  = 2\:\mediumleft(~p\tensor p' + q\tensor q'~\mediumright),\hspace{-1cm}\\[2\jot]
  \mediumleft(~1\tensor 1 - u\tensor u'~\mediumright)
  = 2\:\mediumleft(~p\tensor q' + q\tensor p'~\mediumright).\hspace{-1cm}
\end{align*}
Further one may associate for the degenerate case
\[
  \mediumleft(~1\tensor 1 + u\tensor 1~\mediumright) = 2p\tensor\mediumleft(~p'+q'~\mediumright)
  \TAB\text{and}\TAB
  \mediumleft(~1\tensor 1 - u\tensor 1~\mediumright) = 2q\tensor\mediumleft(~p'+q'~\mediumright)
\]
and similarly for the second tensor factor.\\
On the other hand the spectrum satisfies the obvious relations
\[
  \sigma(a+1)=\sigma(a)+1\TAB\text{and}\TAB\sigma(\lambda a)=\lambda(a).
\]
We may thus combine these two observations to associate binary games
to any order--two polynomial with real coefficients
using the following procedure:

\begin{BREAKPROP}[Real polynomials: binary game]
  Given a polynomial with real possibly mixed coefficients
  \begin{gather*}
    \bias \in\sum\REALS\cdot\{1,u_1,\ldots,u_n\}\tensor\{u_0=1,u_1,\ldots,u_n\}:\\[2\jot]
    \bias = \lambda(1,1)1\tensor1 + \lambda(u_0,u_1) 1\tensor u_1 + \ldots +\lambda(u_n,u_n) u_n\tensor u_n
  \end{gather*}
  one may associate a nonlocal game given by (up to normalisation)
  \begin{gather*}
    \game \sim \bias + \sum \Bigl|\lambda\mediumleft(\{1,u_1,\ldots,u_n\}\times\{u_0=1,u_1,\ldots,u_n\}\mediumright)\Bigr|
  \end{gather*}
  with however generally non-uniform distribution over questions.\\
  Furthermore the constant term above may be generally neglected.
\end{BREAKPROP}

For illustration let us consider the following polynomial
\[
  \bias = \frac{1}{2}~u(x=1)\tensor u(y=2) - \sqrt{2}~u(x=3)\tensor1.
\]
Applying above procedure we obtain the following binary game (up to some normalisation on the distribution of questions)
\begin{gather*}
  \hspace{0.5cm}u(\ldots)=E(+1|\ldots)-E(-1|\ldots):\hspace{1cm}
  \game\sim\bias + \left(\frac{1}{2} + \sqrt{2}\right) =\\[2\jot]
  \begin{aligned}
    =\frac{1}{2}\biggl(E(+1|x=1)\tensor E(+1|y=2) + E(-1|x=1)\tensor E(-1|y=2)\biggr) \\ +
    \sqrt{2}\biggl( E(+1|x=1) + E(-1|x=1)\biggr)\tensor \biggl(E(+1|\mathrm{any})+ E(-1|\mathrm{any})\biggr)
  \end{aligned}
\end{gather*}
where the latter expression may be taken as some question pair
for which the validity of the given answer pair is independent from the second player.\\
As another example we have any signed game as for example those exhibiting a separation between the quantum and quantum commuting value from \cite{DYKEMA-PAULSEN-PRAKASH}.
So any real bias polynomial may be taken as a binary game.

%% file: Quantum-value.tex
\section{Tilted CHSH games: quantum value}
\label{sec:QUANTUM-VALUE}

We now illustrate our solution procedure on the tilted CHSH games:\\
We begin for this with the computation of the quantum value using our machinery as introduced in the previous section:
For this we recall the bias polynomial for the tilted CHSH game from \cite{ACIN-MASSAR-PIRONIO}:
\begin{equation}\label{DEFINITION:TILTED-CHSH}
  \CHSH\mediumleft(\alpha,\beta\in\REALS\mediumright)= \alpha u\tensor(u+v) + v\tensor(u-v) + 2\beta u\tensor1.
\end{equation}
For simplicity of our upcoming formula we have taken a rescaled parameter (note the scaling by a factor two in the final expression). We further usually drop the dependency on the given fixed parameter pair $\alpha,\beta\in\REALS$.

As a first instance we wish to verify that the given polynomial has symmetric spectrum in order to apply our previous machinery:
For this we note that there exists plenty of automorphisms retrieving the negative polynomial, such as for instance by sending the generators for Alice to their negative
\[
  \Phi(\CHSH)=\alpha(-u)\tensor(u+v) + (-v)\tensor(u-v) + 2\beta(-u)\tensor1=-\CHSH
\]
and as such its spectrum is necessarily symmetric as
\[
  \spectrum(\CHSH)=\spectrum\mediumleft(\Phi(\CHSH)\mediumright)=-\spectrum(\CHSH).
\]
In particular we may retrieve the maximal bias as its operator norm
\[
  \|\CHSH\|=\|\sigma(\CHSH)\|=\max\sigma(\CHSH)
\]
and so we may apply our previous machinery as desired.\\
We meanwhile note that in a recent discussion the authors have found an argument for the asymmetry in case of the \mbox{I3322 inequality} of the form
\[
  \spectrum(\mathrm{I3322})\neq -\spectrum(\mathrm{I3322}):\TAB\TAB |\min\spectrum(\mathrm{I3322})|\neq |\max\spectrum(\mathrm{I3322})|
\]
and as such one needs to remain cautious when using above considerations.\\
Further similar transformations reveal the symmetries
\[
  \|\CHSH(-\alpha,\beta)\|=\|\CHSH(\alpha,\beta)\|=\|\CHSH(\alpha,-\beta)\|
\]
and as such we may restrict our attention to $\alpha,\beta\geq0$.\\
With this being said we may now compute its operator norm:

\begin{THM}\label{THEOREM:TILTED:QUANTUM-VALUE}
  The tilted CHSH games have as maximal values ($\alpha,\beta\geq0$)
  \begin{align}
    \Bigl\|\CHSH\mediumleft(~s(\Alice)=0~\medium|~s(\Bob)=\eqref{FORMULA}~\mediumright)\Bigr\|^2 &= 4(1+\alpha^2)(1+\beta^2)
    \label{QUANTUM-VALUE}\tag{Q}\\[2\jot]
      \Bigl\|\CHSH\mediumleft(~s(\Alice)=\pm1~\medium|~s(\Bob)=-1~\mediumright)\Bigr\|^2  &= 4(1+\beta)^2
      \label{LOCAL-VALUE-1}\tag{L1} \\[2\jot]
      \Bigl\|\CHSH\mediumleft(~s(\Alice)=\mathrm{any}~\medium|~s(\Bob)=1~\mediumright)\Bigr\|^2 &= 4(\alpha+\beta)^2
      \label{LOCAL-VALUE-alpha}\tag{La}
  \end{align}
  with optimisation parameter given by the amount of anticommutation
  \[
    \{u\tensor 1,v\tensor 1\}=2s(\Alice)\TAB\text{and}\TAB \{1\tensor u,1\tensor v\}=2s(\Bob)
  \]
  and with first value valid only within the range $|\alpha\beta|\leq1{:}$
  \[
    \hspace{-1.5cm}\eqref{QUANTUM-VALUE}:\TAB
    -1\leq \left(s(\Bob)=\frac{\alpha^2-1}{\alpha^2+1}(\beta^2+1) + \beta^2\right)\leq 1\label{FORMULA}\tag{$\ast$}
  \]
  Furthermore the values above become maximal within the region
  \[
    \begin{gathered}
      \vphantom{\frac{2\beta}{1+\beta^2}}
      \eqref{LOCAL-VALUE-1}\leq\eqref{QUANTUM-VALUE}
      \vphantom{\frac{2\beta}{1+\beta^2}}\\[2\jot]
      \eqref{LOCAL-VALUE-alpha}=\eqref{QUANTUM-VALUE}
    \end{gathered}
    \TAB
    \begin{gathered}
      \vphantom{\frac{2\beta}{1+\beta^2}}\iff\vphantom{\frac{2\beta}{1+\beta^2}}\\[2\jot]
      \iff
    \end{gathered}
    \TAB
    \begin{gathered}
      \alpha^2\geq\frac{2\beta}{1+\beta^2}\\[2\jot]
      |\alpha\beta|=1
    \end{gathered}
  \]
  where the latter local value always lies below the quantum value\\
  whereas their equality coincides precisely with the valid range for \eqref{FORMULA}.
\end{THM}

Before we begin with the proof let us make a few observations:\\
First the region exhibiting a quantum advantage is bounded on one hand by the valid range for the amount of anticommutation given in \eqref{FORMULA}
\[
  \left(s(\Bob)=\frac{\alpha^2-1}{\alpha^2+1}(\beta^2+1) + \beta^2\right)\leq 1
  \TAB\iff\TAB
  \alpha^2\beta^2\leq1
\]
with lower bound always valid and with extreme case given by
\[
  \hspace{-0.5cm}-1=\left(s(\Bob)=\frac{\alpha^2-1}{\alpha^2+1}(\beta^2+1) + \beta^2\right)
  \TAB\iff\TAB
  \alpha=0
\]
while on the other hand by the domain when the quantum value is maximal
\[
  \eqref{LOCAL-VALUE-1}\leq\eqref{QUANTUM-VALUE}
  \TAB\iff\TAB
  \alpha^2\geq\frac{2\beta}{1+\beta^2}
\]
while noting that the other classical value always lies below.\\
Schematically the enclosed region looks like:
\begin{center}
\begin{tikzpicture}[scale=1.6,variable=\x,samples=200,smooth,draw=gray,fill=lightgray!25]
  \fill
    plot[domain=0:1] ({2*\x/(1+\x^2))},\x^2) --
    plot[domain=1:3] (\x,1/\x) --
    plot[domain=3:1] (\x,-1/\x) --
    plot[domain=1:0] ({2*\x/(1+\x^2)},-\x^2);
  \fill
    plot[domain=0:1] (-{2*\x/(1+\x^2))},\x^2) --
    plot[domain=1:3] (-\x,1/\x) --
    plot[domain=3:1] (-\x,-1/\x) --
    plot[domain=1:0] (-{2*\x/(1+\x^2)},-\x^2);
  \begin{scope}[domain=0:sqrt(2)]
    \draw node[anchor=south]{\scriptsize\eqref{LOCAL-VALUE-1}=\eqref{QUANTUM-VALUE}} plot (+{2*\x/(1+\x^2)},+\x^2);
    \draw plot (-{2*\x/(1+\x^2)},+\x^2);
    \draw plot (+{2*\x/(1+\x^2)},-\x^2);
    \draw plot (-{2*\x/(1+\x^2)},-\x^2);
  \end{scope}
  \begin{scope}[domain=1/2:3]
    \draw plot (+\x,+1/\x) node[anchor=south,shift={(0,0)}]{$\scriptstyle s(\Bob)=1$};
    \draw plot (-\x,+1/\x);
    \draw plot (+\x,-1/\x);
    \draw plot (-\x,-1/\x);
  \end{scope}
  \begin{scope}[orange,thick,shorten >=-2.5,shorten <=-2.5,o-o]
    \draw[xshift=0.1] (1,-1) -- (1,1) node[anchor=south west]{$\scriptstyle2\beta\,u\tensor1+\CHSH$};
    \draw[xshift=-0.5] (-1,-1) -- (-1,1) node[anchor=south east]{$\scriptstyle2\beta\,u\tensor1-\CHSH$};
  \end{scope}
  \draw[->] (-3.2, 0) -- (3.2, 0)
    node[near start,anchor=south]{$\scriptstyle\eqref{QUANTUM-VALUE}=\max$}
    node[near end,anchor=south]{$\scriptstyle\eqref{QUANTUM-VALUE}=\max$}
    node[right] {$\scriptstyle\alpha$};
  \draw[->] (0, -2.2) -- (0, 2.2) node[above] {$\scriptstyle\beta$};
\end{tikzpicture}
\end{center}
The orange lines thereby depict the versions appearing in \cite{BAMPS-PIRONIO} and \cite{COLADANGELO-STARK}.\\
In \cite{ACIN-MASSAR-PIRONIO} it has been further stated that one may pass between the range of parameter for the bias polynomial
\[
  (\alpha\leq1) \TAB\leftrightsquigarrow\TAB(\alpha\geq1)
\]
using some transformation on generators as in the beginning of this section.\\
This however cannot be easily achieved by some transformation as above.\linebreak
In fact our characterisation reveals for instance
\[
  \|\CHSH(\alpha,\BLANK)\|\neq\|\CHSH(1/\alpha,\BLANK)\|
\]
or similarly for any other such transformation on the parameter.\\
Instead we have found the more complicated phase separations
\[
  \begin{gathered}
    s(\Bob)\leq1 \\
    \eqref{QUANTUM-VALUE}\geq\eqref{LOCAL-VALUE-1}\vphantom{\frac{2\beta}{1+\beta^2}}
  \end{gathered}
  \TAB
  \begin{gathered}
    \iff \\
    \vphantom{\frac{2\beta}{1+\beta^2}}\iff\vphantom{\frac{2\beta}{1+\beta^2}}
  \end{gathered}
  \TAB
  \begin{gathered}
    |\alpha\beta|\leq1 \\
    \alpha^2\geq\frac{2\beta}{1+\beta^2}
  \end{gathered}
\]
which remained both previously unnoticed.\\
With these observations covered we may now proceed with the proof:

\begin{proof}
  We begin with the following simplification on the first tensor factor:
  \[
    u\tensor1\cdot\CHSH = \alpha\tensor(u+v) + (uv)\tensor(u-v) + 2\beta\tensor1
  \]
  As such we may replace as in the first section $w=uv=z\in\TORUS$:
  \[
    \ldots = \alpha(u+v) + z(u-v) + 2\beta
  \]
  For the second tensor factor we make the familiar observation
  \[
    (u\pm v)^2= u^2\pm\{u,v\} + v^2 = 2(1\pm s)
  \]
  and as such both expressions define themselves unitary generators:\\
  The scaling is more precisely satisfying the constraint
  \[
    0\leq2(1\pm s)\leq 4:\TAB 2(1+s)+2(1-s)=4
  \]
  and as such our unitary generators read for some euclidean optimisation:
  \begin{gather*}
    0\leq x,y\leq 1\TAB\text{with}\TAB x^2+y^2=1:\\[2\jot]
    2xu'=(u+v)\TAB\text{and}\TAB 2yv'= (u-v).
  \end{gather*}
  (The substitutation may be read off from $2(1+s)=4x^2$ and $2(1-s)=4y^2$.)\\
  One could also formulate above generators in more familiar terms as
  \[
    2\cos(\phi)u'=(u+v)\TAB\text{and}\TAB 2\sin(\phi)v'= (u-v)
  \]
  which we however refrain from in favor of the optimisation itself.\\
  We further note that above generators are anticommuting since
  \[
    \{u',v'\}\sim\{u+v,u-v\}=\{u,u\} + \{u,v\} - \{v,u\} - \{v,v\} = 0.
  \]
  As such our bias polynomial from above reads with these generators
  \[
    \ldots =
    2\Bigl( (\alpha x) u' + (yz) v' + \beta \Bigr)
  \]
  where we from now on set for simplicity $u=u'$ and $v=v'$.\\
  This operator fails to be normal (basically because of our very first factorisation) and as such we invoke the \CSTAR-identity in order to access spectral theory which should be always the first step in such investigations.
  As such we consider the squared bias for which we obtain using above expression
  \begin{gather*}
    \frac{1}{4}\mediumleft(\CHSH\mediumright)^2
    = \frac{1}{4}\Bigl(\CHSH\cdot u\tensor1\Bigr)\Bigl(u\tensor1\cdot\CHSH\Bigr) =\hspace{3.5cm} \\[2\jot]
    \begin{aligned}
      &= (\alpha x)^2 + y^2|z|^2 + \beta^2 + \biggl(2(\alpha\beta x) u + (\beta y)(z+z^*) v+ (\alpha xy)(z-z^*)uv\biggr) \\
      &= (\alpha x)^2 + y^2 + \beta^2 + 2\biggl((\alpha\beta x) u + (\beta y)\cos(\omega) v+ (\alpha xy)\sin(\omega)(w:=iuv)\biggr).
    \end{aligned}
  \end{gather*}
  This triplet defines mutually anticommuting generators (they basically define an abstract triplet of Pauli matrices) as one easily infers
  \[
    \{u,v\}=0\TAB\implies\TAB\mediumleft\{~u~\medium|~w=iuv~\mediumright\}=0=\mediumleft\{~v~\medium|~w=iuv~\mediumright\}
  \]
  and as such define selfadjoint unitaries (similar as in our computation above)
  \[
    \mediumleft(a u + bv + c w\mediumright)^2 = a^2u^2 + \ldots + \cancel{2ab \{u,v\}} + \ldots =  a^2+b^2+c^2
  \]
  whence above expression has spectrum given by
  \[
    \spectrum\Bigl(a u + bv + c w\Bigr) = \Bigl\{\pm\sqrt{a^2 + b^2+ c^2}\Bigr\} = \REALS\cap \left\|\smallMATRIX{a\\b\\c}\right\|\TORUS.
  \]
  Together with some constant term we thus obtain as operator norm:
  \[
    \Bigl\|a u + bv + c w +d\Bigr\| = \Bigl\|\spectrum(a u + bv + c w) + d\Bigr\| = \sqrt{a^2 + b^2+ c^2} + |d|
  \]
  This defines some particular instance of random walk behavior!
  For more interesting phenomena of this and similar types we refer to \cite{LEHNER}.\\
  We thus obtain in our case (along our optimisation)
  \begin{gather*}
    \frac{1}{4}\Bigl\|\CHSH\mediumleft(x^2+y^2=1,z=\cos(\omega)+i\sin(\omega)\in\TORUS\mediumright)\Bigr\|^2 =\hspace{2cm} \\[2\jot]
    = (\alpha x)^2 + y^2 + \beta^2 + 2\sqrt{(\alpha\beta x)^2 + (\beta y)^2\cos^2(\omega) + (\alpha xy)^2\sin^2(\omega)}
  \end{gather*}
  which one may now optimise (whose computation we defer to the appendix).\\
  The resulting values and optimal parameter are given in the theorem.
\end{proof}

%% file: Quantum-state.tex
\section{Tilted CHSH games: optimal state}

We may now classify optimal states based on the optimal values we found in the previous section.
In order to handle this task, we first require a more concrete description on the state space in general.
For this we first recall our concrete embedding into the ambient function algebra (from section \ref{sec:REPRESENTATION-THEORY})
\[
  \CSTAR\mediumleft(\INTEGERS/2*\INTEGERS/2\mediumright)=\CSTAR\mediumleft(~u^2=1=v^2~\medium|~\{u,v\}=2\id~\mediumright)\subset\CONTINUOUS\mediumleft(~[-1,1]\to M_2~\mediumright)
\]
from which we obtain for the corresponding twoplayer algebra underlying binary--input binary--output games (confer \cite[section~2]{FREI-CHSH})
\[
  \CSTAR\mediumleft(\twoplayer\mediumright)
  \subset
  \CONTINUOUS\mediumleft(~[-1,1]\times[-1,1]\to M_2\tensor M_2~\mediumright).
\]
We may thus freely extend any state on our two algebra to also act on the ambient function algebra (using Hahn--Banach)
\begin{TIKZCD}
  \CSTAR\mediumleft(\twoplayer\mediumright)\subset \CONTINUOUS\mediumleft(~[-1,1]\times[-1,1]\to M_2\tensor M_2~\mediumright)\rar&\COMPLEX
\end{TIKZCD}
and as such the state space for binary binary--output games in general arises by restriction from the state space for the ambient function algebra:
\[
  \hspace{1.5cm}S\CSTAR\mediumleft(\twoplayer\mediumright) = S\CONTINUOUS\mediumleft([-1,1]\times[-1,1],M_2\tensor M_2\mediumright)
  ~\Big\vert_{ \CSTAR(\twoplayer)}
\]
The latter state space however has a much more tractable state space:\\
The state space admits for instance some familiar description in terms of integration against noncommutative Borel probability measures
\begin{gather*}
    \phi:\CONTINUOUS\Big([-1,1]\times[-1,1],M_2\tensor M_2\Big)\to\COMPLEX:\TAB
    \phi(f)=\int_{[-1,1]\times[-1,1]}f\,\differential\mu
\end{gather*}
which one may derive using the additional embedding $\CONTINUOUS(X,A)\subset\ell^\infty(X,A)$.\\
While this defines a legitimate strategy,
we generally propose any elementary operator algebraic approach
rendering the measure theoretic ones an artefact.\linebreak
Before we do so let us recall that we denote (as is customary) the algebra of continous functions vanishing at infinity as
\[
  \CONTINUOUS_0(X,A)=\CONTINUOUS_0(X)\tensor A
\]
for some locally compact Hausdorff space $X$ and any operator algebra $A$,\\
and recall that the abelian tensor factor defines a nuclear operator algebra.\\
With this in mind we now wish to show the following:

\begin{PROP}[Optimal states: maximising quotient]\label{THEOREM:MAXIMISING-QUOTIENT}
  For any abritrary element (seen as a bias polynomial)
  and with operator norm (seen as its maximal or rather extremal bias)
  \[
    \bias\in\CONTINUOUS_0(X,A):\TAB\TAB \|\bias\|=\sup\|\bias(X)\|
  \]
  consider the corresponding short exact sequence
  \begin{TIKZCD}
    0\rar& \CONTINUOUS_0(\rest,A) \rar& \CONTINUOUS_0(X,A) \rar &\CONTINUOUS(\best,A)\rar& 0
  \end{TIKZCD}
  given by the defining subspace (and its complement)
  \[
    \best:=\Bigl\{~\|\bias(x)\|=\|\bias\|~\Bigr\}\subset X.
  \]
  Then any operator algebraic state $\phi:\CONTINUOUS_0(X,A)\to\COMPLEX$
  maximising above bias in the sense $|\phi(\bias)|=\|\bias\|$
  necessarily factors by the quotient
  \begin{TIKZCD}
    \phi:\CONTINUOUS_0(X,A)
    \rar&
    \CONTINUOUS_0(\best,A)
    \rar&
    \COMPLEX
  \end{TIKZCD}
  and as such the optimal state space arises as a subspace for the quotient.
\end{PROP}
\begin{proof}
  The entire proof rests on the embedding the continuous subalgebra within the larger direct product
  while freely extending any state to act also on the direct product (using Hahn-Banach)
  \begin{TIKZCD}
    \CONTINUOUS_0\mediumleft(X\to A\mediumright)\rar[hookrightarrow]&\prod_XA\rar[dashed]&\COMPLEX
  \end{TIKZCD}
  which has the advantage that indicator functions become accessible.\\
  In particular this allows the description of states as integration against noncommutative measures.
  We will however instead provide an entirely self-contained operator algebraic approach (in particular as the authors share the view of algebraic arguments to be favorable to measure theoretic ones).\\
  For this we may first replace any element by its squared version
  \[
    \bias\TAB\rightsquigarrow\TAB\bias'=\bias^*\bias\geq0.
  \]
  Indeed by the operator algebraic version of Cauchy--Schwarz from \cite{CHOI-1974}
  \[
    |\phi(\bias)|^2=\phi(\bias)^*\phi(\bias)\leq\phi(\bias^*\bias)\leq \|\bias^*\bias\|=\|\bias\|^2
  \]
  from which any optimal state necessarily also maximises the squared bias above.
  We may thus from now on assume the given element to be positive.
  As such we obtain next for the pointwise spectral radius as a continuous function
  \[
    \|\bias(\BLANK)\|\in\CONTINUOUS_0(X):\TAB\TAB
    0\leq \bias(\BLANK)\leq\|\bias(\BLANK)\|\leq\|\bias\|.
  \]
  This bound is operator algebraic since one has the familiar relation
  \[
    \hspace{-1cm}a,b\in{\prod}_{X} A:\TAB\TAB a\leq b\TAB\iff\TAB a(x)\leq b(x)
  \]
  which easily follows using the algebraic definition for their spectrum.\\
  For some given optimal state we thus obtain the inequality
  \[
    \|\bias\|=\phi\Bigl(\bias(\BLANK)\Bigr)\leq\phi\Bigl(\|\bias(\BLANK)\|\Bigr)
    \leq
    \|\bias\|
  \]
  and as such any optimal state necessarily also maximises above pointwise spectral radius as a scalar valued function $\|\bias(\BLANK)\|\in\CONTINUOUS(K)$:
  \[
    \phi\Bigl(\bias(\BLANK)\Bigr)=\|\bias\|\TAB\implies\TAB\phi\Bigl(\|\bias(\BLANK)\|\Bigr)=\|\bias\|.
  \]
  That is above matrix problem infers a corresponding scalar problem as an obstruction for optimal states.
  With this in mind we now aim to verify the scalar analog (of our original problem)
  \[
    \phi\Bigl(\|\bias(\BLANK)\|\Bigr)=\|\bias\|
    \TAB\implies\TAB
    \phi\Bigl(\CONTINUOUS_0(\rest)\Bigr)=0 \TAB ?
  \]
  which we solve using an argument involving hereditary subalgebras:\\
  In order to ease notation let us introduce their difference
  \[
    \phi\biggl(~a(\BLANK):=\|\bias\|-\|\bias(\BLANK)\|\geq0~\biggr)=0.
  \]
  Recall for this the operator algebraic bound (as in any introductory textbook):
  \[
    a^*\Bigl(x^*x\leq\|x^*x\|\Bigr)a
    \TAB\implies\TAB
    \phi\Bigl(a^*x^*xa\Bigr)\leq\|x^*x\|\phi(a^*a)
  \]
  Any operator algebra is however linearly spanned by its positive elements and as such any optimal state also vanishes on the hereditary algebra
  \begin{gather*}
    A=\pos(A)-\pos(A)+i\pos(A)-i\pos(A):\\[0.5\baselineskip]
    \phi\Bigl(a=\sqrt{a}\sqrt{a}\Bigr)=0~\implies~\phi\Bigl(\sqrt{a}A\sqrt{a}\Bigr)=0.
  \end{gather*}
  We meanwhile note that the hereditary subalgebra may be characterised using a corresponding approximate unit and as such the hereditary subalgebra agrees for any positive power (as a well-known result from operator algebras):
  \[
    \overline{aAa}=\mediumleft\{~b~\medium|~a^{1/n}ba^{1/n}\to b~\mediumright\}:\TAB\TAB
    \overline{aAa}=\overline{a^rAa^r}\TAB(r>0).
  \]
  So far we found that any optimal state vanishes on the hereditary subalgebra.\\
  Now any hereditary subalgebra within an abelian operator algebra coincides with its generated ideal
  and as such we are left to identify the expected ideal:\\
  For this we observe their obvious inclusion
  \[
    a\CONTINUOUS_0(X)a\subset\mediumleft\{~b(\best)=0~\mediumright\}\subset\CONTINUOUS_0(\rest).
  \]
  Their equality now easily follows as a trivial application of Stone--Weierstrass:\\
  We note for this that the hereditary subalgebra vanishes nowhere since for
  \begin{gather*}
    \hspace{-2cm}x\in\rest=\mediumleft\{~a(x)>0~\mediumright\}:\TAB\TAB
    a(x)\CONTINUOUS_0(x)a(x)\neq0
  \end{gather*}
  and is nowhere constant since already for $x,y\in\rest$
  \[
    \CONTINUOUS_0(X)x\equiv\CONTINUOUS_0(X)y\TAB\implies x=y
  \]
  and as such the hereditary subalgebra spans the expected ideal as desired.\\
  Note that any subalgebra within the kernel belongs to the multiplicative domain for the operator algebraic state for trivial reason
  \[
    \CONTINUOUS_0(\rest)\subset\ker(\phi)
    \TAB\implies\TAB\phi(ab)=0=\phi(a)\phi(b)\TAB\Bigl(a,b\in\CONTINUOUS_0(\rest)\Bigr)\hspace{-1cm}
  \]
  and as such this extends to above tensor product (see~\cite{CHOI-1974} or \cite{PAULSEN-BOOK})
  \[
    \phi\Bigl(\CONTINUOUS_0(X)\tensor A\Bigr) =
    \phi\Bigl(\CONTINUOUS_0(X)\tensor1\Bigr)\phi\Bigl(1\tensor A\Bigr)=0.
  \]
  Concluding that any optimal state factors by the quotient as desired.
\end{proof}

We may now apply the previous result in case of our tilted CHSH games:\\
Consider for this its maximising subspace as in the previous proposition
\[
  \best:=\Bigl\{~\mediumleft\|\CHSH\mediumleft(~s(\Alice),s(\Bob)~\mediumright)\mediumright\|=\|\CHSH\|~\Bigr\}\subset [-1,1]\times[-1,1].
\]
Since any maximising state necessarily lives on its corresponding quotient
we may thus reduce the problem of optimal states to the quotient:
\begin{gather*}
  \phi\Big(\CHSH(\alpha,\beta)\Big)=\|\CHSH(\alpha,\beta)\|
  \TAB\implies\\[2\jot]
  \begin{tikzcd}[column sep=small]
    \CONTINUOUS\mediumleft(~[-1,1]\times[-1,1]~\medium|~M_2\tensor M_2~\mediumright)
    \rar&
    \CONTINUOUS\mediumleft(~\best~\medium|~M_2\tensor M_2~\mediumright)\rar&\COMPLEX
  \end{tikzcd}
\end{gather*}
In particular we already found the optimal quotient for each region as
\[
  \def\arraystretch{1.5}
  \begin{array}{c|c|c|c}
    \alpha,\beta\in\REALS & \eqref{QUANTUM-VALUE}=\max & \eqref{LOCAL-VALUE-1}=\max & \eqref{LOCAL-VALUE-alpha}=\max \\
    \hline
    \best & \{0\}\times\eqref{FORMULA} & (-1\cup1)\times\{-1\} & [-1,1]\times \{1\}
  \end{array}
\]
whereas the overlapping boundary cases define degenerate cases.\\
We now restrict our attention on the fully noncommutative region given by \eqref{QUANTUM-VALUE}:\\
With above optimal quotient we may thus characterise its optimal states which
we now establish as the second main result of the current section:

\begin{THM}\label{THEOREM:TILTED:QUANTUM-STATE}
  Consider the quantum region (as established in theorem \ref{THEOREM:TILTED:QUANTUM-VALUE})
  \[
    \hspace{-2cm}\eqref{QUANTUM-VALUE}>\eqref{LOCAL-VALUE-1},\eqref{LOCAL-VALUE-alpha}:\TAB\TAB
    \frac{2\beta}{1+\beta^2}<\alpha^2<\frac{1}{\beta^2}
  \]
  and regard its optimal quotient (while recalling theorem \ref{THEOREM:REPRESENTATION-THEORY})
  \[
    M(2)\tensor M(2) =
    \frac{\CSTAR(u^2=1=v^2)}{\{u,v\}=2s(\Alice)}\tensor
    \frac{\CSTAR(u^2=1=v^2)}{\{u,v\}=2s(\Bob)}
  \]
  given by the optimal anticommutation
  \[
    \mediumleft\|\CHSH\mediumleft(~s(\Alice)=0,s(\Bob)=\eqref{FORMULA}~\mediumright)\mediumright\|=\|\CHSH\|.
  \]
  Note that each factor is linearly spanned by a finite number of words in group generators due to each anticommutation $\{u,v\}=2s$
  \begin{align*}
      \TAB \words(u,v)=\{1,u,v,uv=-vu+2s,\ldots\}=\{1,u,v,uv\}
  \end{align*}
  with moreover the exceptional degenerate case $v=\pm u$ whenever $s=\pm1$.\\
  Following this observation any state maximising the quantum value \eqref{QUANTUM-VALUE}
  \[
    \phi:M(2)\tensor M(2)\to\COMPLEX:\TAB\phi\mediumleft(\CHSH\mediumright)=2\sqrt{1+\alpha^2}\sqrt{1+\beta^2}=\mediumleft\|\CHSH\mediumright\|
  \]
  is uniquely and entirely determined by the table of moments (which entails all higher and mixed moments by above observation)
  \[
    \def\arraystretch{1.5}
    \begin{array}[c]{c|c|c|c|c}
      \phi(\ldots) & \BLANK\tensor 1 & \BLANK\tensor 2u & \BLANK\tensor 2v & \BLANK\tensor 4uv \\
      \hline
      1\tensor\BLANK & 1=a^2+d^2 & (a^2-d^2)|w_+| & (a^2-d^2)|w_+| & (a^2+d^2)(|w_+|^2-|w_-|^2) \\
      \hline
      u\tensor\BLANK & a^2-d^2 & (a^2+d^2)|w_+| & (a^2+d^2)|w_+| & (a^2-d^2)(|w_+|^2-|w_-|^2) \\
      \hline
      v\tensor\BLANK & 0 & 2ad|w_-| & -2ad|w_-| & 0 \\
      \hline
      uv\tensor\BLANK & 0 & 0 & 0 & 4ad|w_+|\cdot|w_-|
    \end{array}
  \]
  with familiar transformation satisfying
  \[
    \Bigl|w_+:=u+v\Bigr|=2\alpha\sqrt{\frac{1+\beta^2}{1+\alpha^2}}
    \TAB\text{and}\TAB
    \Bigl|w_-:=u-v\Bigr|=2\sqrt{\frac{1-\alpha^2\beta^2}{1+\alpha^2}}
  \]
  and vector state given by its Schmidt decomposition
  \[
    \ket\phi =
    a\Ket{\vphantom{\tfrac{w_+}{|w_+|}}u\phi_1=\phi_1}\tensor\Ket{\tfrac{w_+}{|w_+|}\psi_1=\psi_1} +
    d\Ket{\vphantom{\tfrac{w_+}{|w_+|}}u\phi_2=-\phi_2}\tensor \Ket{\tfrac{w_+}{|w_+|}\psi_2=-\psi_2}
  \]
  with Schmidt coefficients
  \[
    a=\sqrt{\frac{\sqrt{1+\beta^2}+\beta\sqrt{1+\alpha^2}}{2\sqrt{1+\beta^2}}}
    \TAB\text{and}\TAB
    d=\sqrt{\frac{\sqrt{1+\beta^2}-\beta\sqrt{1+\alpha^2}}{2\sqrt{1+\beta^2}}}
  \]
  from which there is no ambiguity left anymore whatsoever.\\
  The state is maximally entangled only along the horizontal axis
  \[
  \hspace{-1cm}\ket\phi=\maxket:\TAB\TAB
  \mediumleft(~a=\tfrac{1}{\sqrt2}=d~\mediumright)
  \TAB\iff\TAB
  \mediumleft(~\alpha\neq0~\medium|~\beta=0~\mediumright)
  \]
  while always partially entangled (given by the Schmidt-rank)
  \[
    \hspace{-1.7cm}\rank\ket{\phi}=2:\TAB\TAB\mediumleft(~a\neq0~\medium|~d\neq0~\mediumright) \TAB\iff\TAB |\alpha\beta|<1\TAB\checkmark
  \]
  and so defines a nonlocal quantum correlation as desired.\\
  Note that the entire result is in representation--free abstract algebraic form.\\
  So far everything about the unique optimal state for the tilted CHSH games.
\end{THM}

Before we begin with the proof itself we additionally note the following:\\
As in the introductory article \cite{FREI-CHSH}, this entails all quantum commuting correlations
as these arise as states restricted on order--two moments
\[
  \begin{array}{cccc}
    \phi(u\tensor 1), & \phi(v\tensor 1), & \phi(1\tensor u), & \phi(1\tensor v) \\
    \phi(u\tensor u), & \phi(v\tensor u), & \phi(v\tensor u), & \phi(1\tensor v)
  \end{array}
\]
from which one may recover the traditional correlation table given by their corresponding projection valued measures.
In particular more traditional uniqueness of correlations (given for example by self-testing)
is excluding the last row and column from the correlation table given in the theorem above.\\
With this being said let us begin with the proof:

\begin{proof}
  Let us first combine expressions within the tilted CHSH games as
  \[
    \CHSH(\alpha,\beta)= \alpha u\tensor(u+v) + v\tensor(u-v) +2\beta u\tensor1
  \]
  which enforces entirely anticommuting generators as for Alice
  \[
    \{u\tensor 1,v\tensor1\}=0=2s(\Alice)
  \]
  as well as the familiar transformation for the second player
  \[
    \{u+v,u-v\}=\{u,u\}-\{u,v\} + \{v,u\} -\{v,v\}=0
  \]
  which we have also already made use of in the proof of theorem \ref{THEOREM:TILTED:QUANTUM-VALUE}.\\
  As such we obtain an abstract pair of Pauli matrices for each factor
  \begin{gather*}
      u\tensor1=Z\tensor1\TAB\text{and}\TAB v\tensor1=X\tensor1\\[2\jot]
      1\tensor w_+=|w_+|1\tensor Z \TAB\text{and}\TAB 1\tensor w_-=|w_-|1\tensor X
  \end{gather*}
  whose normalisation reads (as in the proof of theorem \ref{THEOREM:TILTED:QUANTUM-VALUE})
  \[
    |w_\pm|^2=(u\pm v)^2=u^2\pm 2s(\Bob)\{u,v\} + v^2 = 2\Bigl(1\pm s(\Bob)\Bigr).
  \]
  Computing those normalisation we obtain (with above pair of Pauli matrices)
  \[
    \CHSH(\alpha,\beta)=2\alpha^2\sqrt{\frac{1+\beta^2}{1+\alpha^2}}Z\tensor Z + 2\sqrt{\frac{1-\alpha^2\beta^2}{1+\alpha^2}}X\tensor X +2\beta Z\tensor1
  \]
  with corresponding maximal bias (from theorem \ref{THEOREM:TILTED:QUANTUM-VALUE})
  \[
    \hspace{-1cm}\eqref{QUANTUM-VALUE}\TAB\TAB\|\CHSH(\alpha,\beta)\|=2\sqrt{1+\alpha^2}\sqrt{1+\beta^2}.
  \]
  As such we aim to solve the following eigenvector problem:
  \begin{gather*}
    \biggl(\alpha^2\sqrt{1+\beta^2}Z\tensor Z + \sqrt{1-\alpha^2\beta^2}X\tensor X +\beta\sqrt{1+\alpha^2} Z\tensor1\biggr)\ket\phi \\
    = \biggl(\lambda:=(1+\alpha^2)\sqrt{1+\beta^2}\biggr)\ket{\phi}
  \end{gather*}
  Note that any operator of the form (as above)
  \[
    T= a Z\tensor Z + bX\tensor X + c Z\tensor 1
  \]
  splits the given Hilbert space into invariant subspaces
  \[
    \COMPLEX\mediumMATRIX{ 1\\0 }\tensor\mediumMATRIX{ 1\\0 } + \COMPLEX\mediumMATRIX{ 0\\1 }\tensor\mediumMATRIX{ 0\\1 }
    \TAB\text{and}\TAB
    \COMPLEX\mediumMATRIX{ 1\\0 }\tensor\mediumMATRIX{ 0\\1 } + \COMPLEX\mediumMATRIX{0\\1}\tensor\mediumMATRIX{ 1\\0 }
  \]
  and as such the eigenvector problem may be individually solved for each invariant subspace.
  On the first invariant subspace the operator acts as
  \begin{gather*}
    \mediumMATRIX{x\\y}:= x \mediumMATRIX{1\\0}\tensor\mediumMATRIX{1\\0} + y \mediumMATRIX{0\\1}\tensor\mediumMATRIX{0\\1}:\\[2\jot]
    T\mediumMATRIX{x\\y}= \biggl(a1 + bX + cZ\biggr)\mediumMATRIX{x\\y}
  \end{gather*}
  while on the second one as
  \begin{gather*}
    \mediumMATRIX{x\\y}:= x \mediumMATRIX{1\\0}\tensor\mediumMATRIX{0\\1} + y \mediumMATRIX{0\\1}\tensor\mediumMATRIX{1\\0}:\\[2\jot]
    T\mediumMATRIX{x\\y}= \biggl(-a1 + bX + cZ\biggr)\mediumMATRIX{x\\y}.
  \end{gather*}
  As such the eigenspace problem from above reads on respective subspace:
  \[
    \biggl(bX + cZ\biggr)\mediumMATRIX{x\\y}=(\lambda\mp a)\mediumMATRIX{x\\y}\TAB?
  \]
  Note that (as also several times before) the left-hand side in both cases defines a selfadjoint unitary (up to normalisation)
  and as such a necessary condition for the existence of nontrivial solutions becomes respectively
  \[
    (bX + cZ)^2= b^2+c^2 = (\lambda\mp a)^2 \TAB?
  \]
  This necessary condition is satisfied for the first subspace
  \[
    \hspace{1cm}b^2+c^2=(1-\alpha^2\beta^2) + \beta^2(1+\alpha^2) =
    (1+\beta^2)=(\lambda-a)^2\TAB\checkmark
  \]
  whereas this fails for the second subspace (unless $\alpha=0$)
  \[
    \hspace{0.5cm}b^2+c^2=(1+\beta^2) \neq
    (1+2\alpha^2)^2(1+\beta^2) = (\lambda+a)^2\TAB\times
  \]
  noting that the exceptional case lies outside the desired region
  $\eqref{LOCAL-VALUE-1},\eqref{LOCAL-VALUE-alpha}\,{<}\,\eqref{QUANTUM-VALUE}$.\\
  As such the eigenspace is one-dimensional and so there is a unique optimal state given by the solution to the equation for the first invariant subspace.

  We may now solve the eigenvector problem for the valid first invariant subspace given by $bX\smallMATRIX{x\\y}=\mediumleft( (\lambda-a)1 + cZ\mediumright)\smallMATRIX{x\\y}$ which reads written out
  \[
    \begingroup
    \renewcommand*{\arraystretch}{1.5}
    \MATRIX{ \sqrt{1-\alpha^2\beta^2}~y\\ \sqrt{1-\alpha^2\beta^2}~x } =
    \MATRIX{ \mediumleft(\sqrt{1+\beta^2}-\beta\sqrt{1+\alpha^2}\mediumright)~x \\
    \mediumleft(\sqrt{1+\beta^2}+\beta\sqrt{1+\alpha^2}\mediumright)~y }
    \endgroup
  \]
  while further noting the helpful identity (which is not too evident at first)
  \[
    \mediumleft(1-\alpha^2\beta^2\mediumright) = \Bigl(\sqrt{1+\beta^2}+\beta\sqrt{1+\alpha^2}\Bigr)\Bigl(\sqrt{1+\beta^2}-\beta\sqrt{1+\alpha^2}\Bigr).
  \]
  As such above equation reduces to the following more symmetric relation:
  \[
    \sqrt{\sqrt{1+\beta^2}-\beta\sqrt{1+\alpha^2}}~x = \sqrt{\sqrt{1+\beta^2}+\beta\sqrt{1+\alpha^2}}~y
  \]
  Note these define the final Schmidt coefficients as
  \begin{gather*}
    \ket\varphi = \mediumMATRIX{x\\y}
    = x\mediumMATRIX{1\\0}\tensor\mediumMATRIX{1\\0} +
    + y\mediumMATRIX{0\\1}\tensor\mediumMATRIX{0\\1}\\
    = x\Ket{\vphantom{\tfrac{w_+}{|w_+|}}u\varphi_1=\varphi_1}\tensor\Ket{\tfrac{w_+}{|w_+|}\psi_1=\psi_1} +
    + y\Ket{\vphantom{\tfrac{w_+}{|w_+|}}u\varphi_2=-\varphi_2}\tensor \Ket{\tfrac{w_+}{|w_+|}\psi_2=-\psi_2}
  \end{gather*}
  and as such above equation reveals the ratio of Schmidt coefficients,\\
  whereas their normalisation follows from the normalisation
  \[
    x^2+y^2 =
    \Bigl(\sqrt{1+\beta^2}+\beta\sqrt{1+\alpha^2}\Bigr) +
    \Bigl(\sqrt{1+\beta^2}-\beta\sqrt{1+\alpha^2}\Bigr) = 2\sqrt{1+\beta^2}.
  \]
  From these we may finally derive every moment such as for instance
  \begin{gather*}
    \bra\phi 1\tensor 2u \ket\phi = \bra\phi 1\tensor(w_+ + w_-)\ket\phi =\hspace{2cm}\\[2\jot]
    \hspace{2cm}=\bra\phi \biggl( |w_+|1\tensor Z + |w_-|1\tensor X\biggr)\ket\phi=(a^2-d^2)|w_+|
  \end{gather*}
  which may be easily read off from
  \begin{gather*}
    \biggl(a({\scriptstyle 1\:0})\tensor({\scriptstyle 1\:0}) + d({\scriptstyle 0\:1})\tensor({\scriptstyle 0\:1})\biggr)
    1\tensor Z
    \biggl(a\smallMATRIX{1\\0}\tensor\smallMATRIX{1\\0} + d\smallMATRIX{0\\1}\tensor\smallMATRIX{0\\1}\biggr) = (a^2-d^2) \\
    \biggl(a({\scriptstyle 1\:0})\tensor({\scriptstyle 1\:0}) + d({\scriptstyle 0\:1})\tensor({\scriptstyle 0\:1})\biggr)
    1\tensor X
    \biggl(a\smallMATRIX{1\\0}\tensor\smallMATRIX{1\\0} + d\smallMATRIX{0\\1}\tensor\smallMATRIX{0\\1}\biggr) =0
  \end{gather*}
  whereas we obtain for higher moments such as
  \[
    \bra\phi v\tensor uvu \ket\phi
    = -\bra\phi v\tensor v \ket\phi + 2s(\Bob)\bra\phi v\tensor u\ket\phi
    =2ad|w_+|\Bigl(1+2s(\Bob)\Bigr).
  \]
  As such we have found the entire information on the unique optimal state.\\
  This concludes the proof of the second main result from this section.
\end{proof}

%% file: Acknowledgements.tex
\section*{Acknowledgements}

The first named author would like to thank Remigiusz Augusiak for pointing out the quite valuable perspective of partial entanglement for correlations.\\
The second named author would like to further acknowledge the Bad Hoennef Summer school on Quantum Computing
for some quite valuable insights.\\
Furthermore the authors would like to thank Moritz Weber for the kind hospitality during the Focus semester on Quantum Information Saarbruecken.\\
Finally the first named author acknowledges the support under the\linebreak \mbox{Marie Curie} Doctoral~Fellowship~No.~801199.

%% file: Optimisation.tex
\section{Optimisation problem}

We now solve the optimisation problem from the proof of theorem \ref{THEOREM:TILTED:QUANTUM-VALUE}:\\
Recall for this the optimisation parameter (in terms of generating unitaries)
\begin{gather*}
  \begin{gathered}
    \Alice:\\
    \Bob:
  \end{gathered}
  \TAB\TAB
  \begin{gathered}
    w=uv\tensor1= \cos(\omega)+i\sin(\omega)\in\TORUS\\
    4(x^2+y^2)=1\tensor|u+v|^2 + 1\tensor|u-v|^2=4
  \end{gathered}
  \hspace{1cm}
\end{gather*}
from which the optimisation problem reads:
\begin{gather*}
  \frac{1}{4}\Bigl\|\CHSH\mediumleft(w=\cos(\omega)+i\sin(\omega)\medium|x^2+y^2=1\mediumright)\Bigr\|^2 =\hspace{2cm}\\
  = (\alpha x)^2 + y^2 + \beta^2 + 2\sqrt{(\alpha\beta x)^2 + (\beta y)^2\cos^2(\omega) + (\alpha xy)^2\sin^2(\omega)}
\end{gather*}
We may conveniently combine above expression using $\mediumleft(A=\alpha^2x^2\medium|B=y^2\mediumright)$:
\[
  = A + B + \beta^2 + 2\sqrt{\beta^2(A+B)\cos^2(\omega)+A(B+\beta^2)\sin^2(\omega)}
\]
We note that the expression outside the square root does not depend on the optimisation parameter for the first player Alice.
As such it suffices to consider the boundary cases for said optimisation parameter:
\begin{align*}
  \cos^2(\omega)=1:\TAB\TAB
  &= (A + B) + \beta^2 + 2\beta\sqrt{(A+B)} = \Bigl(\sqrt{A+B}+\beta^2\Bigr)^2\\
  \sin^2(\omega)=1:\TAB\TAB
  &= A + (B + \beta^2) + 2\sqrt{A(B+\beta^2)} = \Bigl(\sqrt{A}+\sqrt{B+\beta^2}\Bigr)^2
\end{align*}
The first case is easily solved and reveals the values \eqref{LOCAL-VALUE-1} and \eqref{LOCAL-VALUE-alpha}:
\[
  \max\Bigl(\sqrt{A+B}+\beta\Bigr)^2
  = \Bigl(\sqrt{\max(A+B)}+\beta\Bigr)^2
  = \Bigl(\max(\alpha,1)+\beta\Bigr)^2.
\]
The second case reveals the interesting quantum value.
For this we search for local maxima in terms of its derivative which reads in our case
\[
  d\Bigl(\sqrt{A}+\sqrt{B+\beta^2}\Bigr) =
  \frac{dA}{2\sqrt{A}} + \frac{dB}{2\sqrt{B+\beta^2}} = 0\TAB{?}
\]
In particular we obtain as a necessary squared relation:
\[
  dA\sqrt{B+\beta^2} = -dB\sqrt{A}
  \TAB\implies\TAB
  (dA)^2(B+\beta^2) = (dB)^2A
\]
In our case those derivatives read using $x^2+y^2=1$,
\[
  \frac{dA}{d(r=x^2)}=\alpha^2
  \TAB\text{and}\TAB
  \frac{dB}{d(r=x^2)}=-1.
\]
As such the solution for the squared relation also carries the correct sign.\\
Inserting those in above squared relation reveals as optimal relation
\begin{gather*}
  \left( \frac{dA}{dr} \right)^2 (B+\beta^2) = \left( \frac{dB}{dr} \right)^2 A
  \TAB\implies\TAB
  x^2=\alpha^2\frac{1+\beta^2}{1+\alpha^2}
\end{gather*}
and one may indeed verify that this solves above unsquared equation.\\
Inserting the found relation we obtain as maximal quantum value
\[
  \TAB=\Bigl(\sqrt{A}+\sqrt{B+\beta^2}\Bigr)^2 = (1+\alpha^2)(1+\beta^2)=\frac{1}{4}\eqref{QUANTUM-VALUE}\TAB\TAB\checkmark
\]
which recovers the familiar quantum value as desired.
We finally note how above relation agrees with the amount of anticommutation between generators
\begin{gather*}
  4x^2= 1\tensor|u+v|^2 = 2\mediumleft(1+\{1\tensor u,1\tensor v\}\mediumright)=2\mediumleft(1+s(\Bob)\mediumright):\\[0.5\baselineskip]
  2\mediumleft(s(\Bob)+1\mediumright)= 2\left(\frac{\alpha^2-1}{\alpha^2+1}(\beta^2+1)+(\beta^2+ 1)\right)=4\alpha^2\frac{1+\beta^2}{1+\alpha^2}=4x^2\TAB\checkmark\hspace{-0.5cm}
\end{gather*}
and as such we also recovered the stated amount of anticommutation \eqref{FORMULA}.\\
This concludes the solution to our optimisation problem from theorem \ref{THEOREM:TILTED:QUANTUM-VALUE}.